\documentclass{myjournal}

\usepackage[a4paper, margin=1in]{geometry}

\usepackage{cite}
\usepackage{hyperref}
\usepackage{caption}
\usepackage{subcaption}
\usepackage{amsmath}
\usepackage{placeins}

\usepackage{amsthm}
\newtheorem{theorem}{Theorem}
\newtheorem{corollary}[theorem]{Corollary}

\newcommand{\maxline}[1]{\overline{#1}}
\usepackage[colorinlistoftodos]{todonotes}

\newcommand{\doneupdate}[1]{#1}
\newcommand{\todolater}[1]{}

\newcommand{\fprint}{f}
\newcommand{\quot  }{\textit{quot}}
\newcommand{\rem   }{\textit{rem}}
\newcommand{\readlock }{\texttt{010}}
\newcommand{\writelock}{\texttt{110}}
\newcommand{\fprate}{p^{+}}
\newcommand{\fpratei}[1]{p^+_{#1}}
\newcommand{\fpratedach}{\maxline{p}^{+}}
\newcommand{\fprateidach}[1]{\maxline{p}^+_{#1}}
\newcommand{\filldeg}{\delta}

\definecolor{lpcol}{HTML}{0D5191}
\definecolor{llqfcol}{HTML}{4FB743}
\definecolor{elqfcol}{HTML}{F28123}
\definecolor{cqfcol}{HTML}{F54080}
\definecolor{bloomcol}{HTML}{950952}

\usepackage{amssymb}
\usepackage{latexsym}
\usepackage{pifont}
\usepackage{graphicx}
\newcommand{\sqrfull}{{\color{lpcol}%
\scalebox{0.95}{\ding{110}}%
}}

\newcommand{\cirfull}{{\color{llqfcol}%
\ding{108}%
}}

\newcommand{\diafull}{{\color{elqfcol}%
\scalebox{0.9}{\rotatebox[origin=c]{45}{\ding{110}}}%
}}

\newcommand{\cqfmarker}{{\color{cqfcol}%
\scalebox{1.15}{\rotatebox[origin=c]{45}{\ding{58}}}%
}}

\newcommand{\trifull}{{\color{bloomcol}%
\ding{115}%
}}

\usepackage[ruled]{algorithm}
\newenvironment{code}{\vspace{-0.5em}\noindent%
\begin{tabbing}%
\hspace{1.5em}\=\hspace{1.5em}\=\hspace{1.5em}\=\hspace{1.5em}\=\hspace{1.5em}\=\hspace{1.5em}\=\hspace{1.5em}\=%
\hspace{1.5em}\=\hspace{1.5em}\=\hspace{1.5em}\=\hspace{1.5em}\=\hspace{1.5em}\=%
\kill}{\end{tabbing}\vspace{-0.5em}}

\definecolor{ccommentcol}{HTML}{A0A0A0}
\definecolor{cemphcol}{HTML}{4047A0}
\definecolor{cvarcol}{HTML}{000000}
\newcommand{\ccomment}[1]{{\color{ccommentcol} \tt// #1}}
\newcommand{\cvar}[1]{{\textit{\color{cvarcol} #1}}}
\newcommand{\crunvar}{\cvar{it}}
\newcommand{\cscanl}{{\bf scan left from}}
\newcommand{\cscanr}{{\bf scan right from}}
\newcommand{\cif}{{\bf if}}
\newcommand{\cand}{{\bf and}}
\newcommand{\cthen}{{\bf then}}
\newcommand{\cbreak}{{\bf break}}
\newcommand{\ccontinue}{{\bf continue}}
\newcommand{\cwait}{{\bf wait}}
\newcommand{\clock}{{\bf lock}}
\newcommand{\cunlock}{{\bf unlock}}
\newcommand{\creturn}{{\bf return}}
\newcommand{\cemph}[1]{{\sf \color{cemphcol} #1}}

\begin{document}

\title{Concurrent Expandable AMQs on the Basis of Quotient Filters}
\author{Tobias Maier, Peter Sanders, Robert Williger}
\institute{Karlsruhe Institute of Technology\\\email{{\{t.maier, sanders\}@kit.edu}}}

\date{}

\maketitle


\begin{abstract}
  A quotient filter is a cache efficient Approximate Membership Query
  (AMQ) data structure.  Depending on the fill degree of the filter
  most insertions and queries only need to access one or two
  consecutive cache lines.  This makes quotient filters very fast
  compared to the more commonly used Bloom filters that incur multiple
  cache misses depending on the false positive rate.  However,
  concurrent Bloom filters are easy to implement and can be
  implemented lock-free while concurrent quotient filters are not as
  simple.  Usually concurrent quotient filters work by using an
  external array of locks -- each protecting a region of the table.
  Accessing this array incurs one additional cache miss per operation.
  We propose a new locking scheme that has no memory overhead.  Using
  this new locking scheme we achieve \doneupdate{1.8} times higher
  speedups than with the common external locking scheme.

  Another advantage of quotient filters over Bloom filters is
  that a quotient filter can change its size when it is becoming full.
  We implement this growing technique for our concurrent quotient
  filters and adapt it in a way that allows unbounded growing while
  keeping a bounded false positive rate.  We call the resulting data
  structure a fully expandable quotient filter.  Its design is similar to
  scalable Bloom filters, but we exploit some concepts inherent to
  quotient filters to improve the space efficiency and the
  query speed.

  Additionally, we propose several quotient filter variants
  that are aimed to reduce the number of status bits (2-status-bit
  variant) or to simplify concurrent implementations (linear probing
  quotient filter).  The linear probing quotient filter even leads to
  a lock-free concurrent filter implementation.  This is especially
  interesting, since we show that any lock-free implementation of
  another common quotient filter variant would incur significant
  overheads in the form of additional data fields or multiple passes
  over the accessed data.
\end{abstract}

\section{Introduction}
\label{sec:intro}


\emph{Approximate Membership Query} (AMQ) data structures offer a
simple interface to represent sets.  Elements can be inserted,
queried, and depending on the use case elements can also be removed.  A
query returns true if the element was previously inserted.  Querying
an element might return true even if an element was not inserted.  We
call this a false positive.  The probability that a non-inserted
element leads to a false positive is called the false positive rate of
the filter. It can usually be chosen at the time of constructing the
data structure.  AMQ data structures have two main advantages over
other set representations they are fast and space efficient.  Similar
to hash tables most AMQ data structures have to be initialized knowing
a bound to their final size.  This can be a problem for their space
efficiency if no accurate bound is known.

AMQ data structures have become an integral part of many complex data
structures or data base applications.  Their small size and fast
accesses can be used to sketch large, slow data sets.  The AMQ filter
is used before accessing the data base to check whether a slow lookup
is actually necessary.  Some recent uses of AMQs include network
analysis~\cite{qf_network} and bio informatics~\cite{qf_bio1}.  These
are two of the areas with the largest data sets available today, but
given the current big data revolution we expect more and more research
areas will have a need for the space efficiency and speedup potential
of AMQs.  The most common AMQ data structure in practice is still a
Bloom filter even though a quotient filter would usually be
significantly faster.  One reason for this might be that concurrent
Bloom filters are easy to implement -- even lock-free and well scaling
implementations.  This is important because well scaling
implementations have become critical in today's multi-processor
scenarios since single core performance stagnates and efficient
multi-core computation has become a necessity to handle growing data
sets.

We present a technique for concurrent quotient filters that uses local
locking inside the table -- without an external array of locks.
Instead of traditional locks we use the per-slot status bits that are
inherent in quotient filters to lock local ranges of slots.  As long
as there is no contention on a single area of the table there is very
little overhead to this locking technique because the locks are placed
in slots that would have been accessed anyways.
An interesting advantage of quotient filters over Bloom filters is
that the capacity of a quotient filter can be increased without access
to the original elements.  However, because the false positive rate
grows linearly with the number of inserted elements, this is only useful for
a limited total growing factor.  We implement this growing technique
for our concurrent quotient filters and extend it to allow large
growing factors at a strictly bounded false positive rate. These
\emph{fully expandable quotient filters} combine growing quotient filters
with the multi level approach of scalable Bloom
filters~\cite{scalable_bloom}.

\subsection{Related Work}
\label{sec:intro_related}
Since quotient filters were first described in
2012~\cite{dont_thrash} there has been a steady stream of
improvements.  For example Pandey et al.~\cite{counting_qf} have
shown how to reduce the memory overhead of quotient filters by using
rank-select data structures.  This also improves the performance when
the table becomes full. Additionally, they show an idea that saves
memory when insertions are skewed (some elements are inserted many
times).  They also mention the possibility for concurrent access using
an external array of locks (see Section~\ref{sec:exp} for results).
Recently, there was also a GPU-based implementation of quotient
filters~\cite{gpu_qf}.  This is another indicator that there is a lot
of interest in concurrent AMQs even in these highly parallel
scenarios.

Quotient filters are not the only AMQ data structures that have
received attention recently. Cuckoo
filters~\cite{cuckoo_filter,adaptive_cuckoo_filter} and very recently
Morton filters~\cite{morton_filter} (based on cuckoo filters) are two
other examples of AMQ data structures.  A cuckoo filter can be filled
more densely than a quotient filter, thus, it can be more space
efficient.  Similar to cuckoo hash tables each access to a cuckoo
filter needs to access multiple buckets -- leading to a worse
performance when the table is not densely filled (a quotient filter
only needs one access). It is also difficult to implement concurrent
cuckoo filters (similar to concurrent cuckoo hash
tables~\cite{concurrent_cuckoo, lockfree_cuckoo}); in particular,
when additional memory overhead is unwanted, e.g., per-bucket locks.

A scalable Bloom filter~\cite{scalable_bloom} allows unbounded
growing.  This works by adding additional levels of Bloom filters once
one level becomes full.  Each new filter is initialized with a higher
capacity and a lower false positive rate than the last.  The query
time is dependent on the number of times the filter has grown (usually
logarithmic in the number of insertions).  Additionally, later filters
need more and more hash functions to achieve their false positive rate
making them progressively slower.  In Section~\ref{sec:dyn_quotient},
we show a similar technique for fully expandable quotient filters that
mitigates some of these problems.

\subsection{Overview}
\label{sec:intro_overview}
In Section~\ref{sec:prelim} we introduce the terminology and explain
basic quotient filters as well as our variations in a sequential
setting.  Section~\ref{sec:conc} then describes our approach to
concurrent quotient filters including the two variants: the lock-free
linear probing quotient filter and the locally locked quotient filter.
The dynamically growing quotient filter variant is described in
Section~\ref{sec:dyn_quotient}.  All presented data structures are
evaluated in Section~\ref{sec:exp}.  Afterwards, we draw our
conclusions in Section~\ref{sec:dis}.

\section{Sequential Quotient Filter}
\label{sec:prelim}
In this section we describe the basic sequential quotient filter as
well as some variants to the main data structure.  Throughout this
paper, we use $m$ for the number of slots in our data structure and
$n$ for the number of inserted elements.  The fill degree of any given
table is denoted by $\filldeg = n/m$.  Additionally, we use $\fprate$
to denote the probability of a false positive query.

\subsection{Basic Quotient Filter}
\label{sec:prelim_quotient}
Quotient filters are approximate membership query data structures that
were first described by Bender et al.~\cite{dont_thrash} and build on
an idea for space efficient hashing originally described by
Cleary~\cite{Cle84}.  Quotient filters replace possibly large elements
by \emph{fingerprints}.  The fingerprint $\fprint(x)$ of an element
$x$ is a number in a predefined range
$\fprint:\ x\mapsto \{0,...,2^k -1\}$ (binary representation with
exactly $k$ digits).
We commonly get a fingerprint of $x$ by taking the $k$ bottommost bits
of a hash function value $h(x)$ (using a common hash function like
xxHash~\cite{xxhash}).

A quotient filter stores the fingerprints of all inserted elements.
When executing a query for an element $x$, the filter returns
\texttt{true} if the fingerprint $\fprint(x)$ was previously inserted
and \texttt{false} otherwise.  Thus, a query looking for an element
that was inserted always returns \texttt{true}.  A false positive
occurs when $x$ was not inserted, but its fingerprint $\fprint(x)$
matches that of a previously inserted element.  Given a fully random
fingerprint function,
the probability of two fingerprints being the same is $2^{-k}$.
Therefore, the probability of a false positive is bounded by
$n\cdot 2^{-k}$ where $n$ is the number of stored fingerprints.

To achieve expected constant query times as well as to save memory,
fingerprints are stored in a special data structure that is similar to
a hash table with open addressing (specifically it is similar to
robinhood hashing~\cite{robin}).  During this process, the fingerprint
of an element $x$ is split into two parts: the topmost $q$ bits called
the quotient $\quot(x)$ and the bottommost $r$ bits called the
remainder $\rem(x)$ ($q+r=k$).  The quotient is used to address a
table that consisting of $m=2^q$ memory slots of $r+3$ bits (it can
store one remainder and three additional status bits).  The quotient
of each element is only stored implicitly by the position of the
element in the table. The remainder is stored explicitly within one
slot of the table.  Similar to many hashing techniques, we try to
store each element in one designated slot which we call its canonical
slot (index $\quot(x)$).  With the help of the status bits we can
reconstruct the quotient of each element even if it is not stored in
its canonical slot.

\begin{table*}[htb]
\noindent
\begin{minipage}[b]{0.32\textwidth}
  \begin{tabular}{l|l}
    \texttt{1**}& this slot has a run\\
    \texttt{000}& empty slot\\
    \texttt{100}& cluster start\\
    \texttt{*01}& run start\\
    \texttt{*11}& continuation of a run\\
    \texttt{*10}& -- (not used see~\ref{sec:conc_status_locking})\\
  \end{tabular}
\vspace{2mm}
  \caption{\label{tab:status} Meaning of
    different status bit combinations.}
\end{minipage}\hfill
\begin{minipage}[b]{.66\textwidth}
  \renewcommand\tablename{Figure}
  \setcounter{table}{0}
  \includegraphics[width=\textwidth]{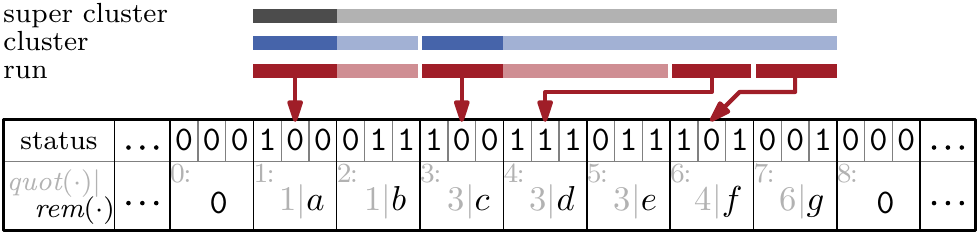}
  \caption{\label{fig:cluster} Section of the table with
    highlighted runs, clusters, and superclusters.  Runs point to
    their canonical slot.}
  \renewcommand\tablename{Table}
  \setcounter{figure}{1}
  \setcounter{table}{1}
\end{minipage}
\end{table*}

The main idea for resolving collisions is to find the next free slot
-- similar to linear probing hash tables -- but to reorder the elements such that they
are sorted by their fingerprints (see Figure~\ref{fig:cluster}).
Elements with the same quotient (the same canonical slot) are stored
in consecutive slots, we call them a run.  The canonical run of an
element is the run associated with its canonical slot.  The canonical
run of an element does not necessarily start in its canonical slot.  It
can be shifted by other runs.  If a run starts in its canonical slot
we say that it starts a cluster that contains all shifted runs after
it.  Multiple clusters that have no free slots between them form a
supercluster.
We use the 3 status bits that are part of each slot to distinguish
between runs, clusters, and empty slots. For this we store the
following information about the contents of the slot (further
described in Table~\ref{tab:status}): Were elements hashed to this
slot (is its run non-empty)? Does the element in this slot belong to
the same run as the previous entry (used as a run-delimiter signaling
where a new run starts)?  Does the element in this slot belong to the
same cluster as the previous entry (is it shifted)?

During the query for an element $x$, all remainders stored in $x$'s
canonical run have to be compared to $\rem(x)$.  First we look at the
first status bit of the canonical slot.  If this status bit is unset
there is no run for this slot and we can return false.  If there is a
canonical run, we use the status bits to find it by iterating to the
left until the start of a cluster (third status bit $= 0$).  From
there, we move to the right counting the number of non-empty runs to
the left of the canonical run (slots with first status bit $=1$ left
of $\quot(x)$) and the number of run starts (slots with second status
bit $=0$).  This way, we can easily find the correct run and compare
the appropriate remainders.

An insert operation for element $x$ proceeds similar to a query, until
it either finds a free slot or a slot that contains a fingerprint that
is $\geq\fprint(x)$.  The current slot is replaced with $\rem(x)$
shifting the following slots to the right (updating the status bits
appropriately).

When the table fills up, operations become slow because the average
cluster length increases.  When the table is full, no more insertions
are possible.  Instead, quotient filters can be migrated into a larger
table -- increasing the overall capacity.  To do this a new table is
allocated with twice the size of the original table and one less bit
per remainder.  Addressing the new table demands an additional
quotient bit, since it is twice the size.  This issue is solved by
moving the uppermost bit from the remainder into the quotient
($q' = q+1$ and $r' = r-1$).  The capacity exchange itself does not
impact the false positive rate of the filter (fingerprint size and
number of elements remain the same).  But the false positive rate
still increases linearly with the number of insertions
($\fprate{} = n\cdot 2^{-k}$).  Therefore, the false positive rate
doubles when the table is filled again.  We call this migration
technique \emph{bounded growing}, because to guarantee reasonable
false positive rates this method of growing should only be used a
bounded number of times.

\subsection{Variants}
We give a short explanation of counting quotient filters using
rank-and-select data structures~\cite{counting_qf}. Additionally, we
introduce two previously unpublished quotient filter variants: the
2-status-bit quotient filter and the linear probing quotient filter.

\paragraph*{(Counting) Quotient Filter Using Rank-and-Select.}
\label{sec:var_rank_select}
This quotient filter variant proposed by Pandey et
al.~\cite{counting_qf} introduces two new features: the number of
necessary status bits per slot is reduced to 2.125 by using a
bitvector that supports rank-and-select queries and the memory
footprint of multisets is reduced by allowing counters to be stored in
the table -- together with remainders.

This quotient filter variant uses two status bits per slot the
\emph{occupied-bit} (first status bit in the description above) and
the \emph{run-end-bit} (similar to the second bit).  These are stored
in two bitvectors that are split into 64~bit blocks (using 8
additional bits per block).  To find a run, the rank of its
occupied-bit is computed and the slot with the matching run-end-bit is
selected.

To reduce memory impact of repeated elements Pandey et al. use the
fact that remainders are stored in increasing order.  Whenever an
element is inserted for a second time instead of doubling its
remainder they propose to store a counter in the slot after its
remainder this counter can be recognized as a counter because it does
not fit into the increasing order (the counter starts at zero).  As
long as it remains smaller than the remainder this counter can count
the number of occurrences of this remainder.  Once the counter becomes
too large another counter is inserted.

\paragraph*{Two-Status-Bit Quotient Filter.}
\label{sec:var_2bit}
Pandey et al.~\cite{counting_qf} already proposed a 2 status bit
variant of their counting filter.  Their implementation however has
average query and insertion times in $\Theta(n)$.  The goal of our
2-status bit variant is to achieve running times close to
$O(\textit{supercluster length})$.  To achieve this, we change the
definition of the fingerprint (only for this variant) such that no
remainder can be zero
$\fprint{}': x\mapsto \{0, ..., 2^q-1\}\times\{1, ..., 2^r-1\}$.
Obtaining a non-zero remainder can easily be achieved by rehashing an
element with different hash functions until the remainder is non-zero.
This change to the fingerprint only has a minor impact on the false
positive rate of the quotient filter ($n/m\cdot (2^r-1)$ instead of
$n/m\cdot 2^r$).

Given this change to the fingerprint we can easily distinguish empty
slots from filled ones.  Each slot uses two status bits the
\emph{occupied-bit} (first status bit in the description above) and
the \emph{new-run-bit} (run-delimiter).  Using these status bits we
can find a particular run by going to the left until we find a free
slot and then counting the number of occupied- and new-run-bits while
moving right from there (\emph{Note}: a cluster start within a larger
supercluster cannot be recognized when moving left).

\paragraph*{Linear Probing Quotient Filter.}
\label{sec:var_lpqf}
This quotient filter is a hybrid between a linear
probing hash table and a quotient filter.  It uses no reordering of
stored remainders and no status bits.  To offset the removal of status
bits we add three additional bits to the remainder (leading to a
longer fingerprint).  Similar to the two-status-bit quotient
filter above we ensure no element $x$ has $\rem(x) = 0$ (by
adapting the fingerprint function).

During the insertion of an element $x$ the remainder $\rem(x)$ is
stored in the first empty slot after its canonical slot.  Without
status bits and reordering it is impossible to reconstruct the
fingerprint of each inserted element. Therefore, a query looking for
an element $x$ compares its remainder $\rem(x)$ with every remainder
stored between $x$'s canonical slot and the next empty slot.  There
are potentially more remainders that are compared than during the same
operation on a normal quotient filter.  The increased remainder length
however reduces the chance of a single comparison to lead to a false
positive by a factor of 8. Therefore, as long as the average number of
compared remainders is less than 8 times more than before, the false
positive remains the same or improves.  In our tests, we found this to
be true for fill degrees up to $\approx 70\%$.  The number of compared
remainders corresponds to the number of slots probed by a linear
probing hash table, this gives the bound below.

\begin{corollary}[$\fprate$ linear probing q.f.]
  The false positive rate of a linear probing quotient filter with
  $\filldeg = n/m$ (bound due to Knuth~\cite{knuth}
  chapter 6.4) is\\

  \vspace{-6mm}\footnotesize \[\fprate{} = \frac{E[\#\textit{comparisons}]}{2^{r+3} -1} =
      \frac12\left(1+\frac{1}{(1-\filldeg)^2}\right)\cdot\frac{1}{2^{r+3} -1}\]
\end{corollary}

It should be noted that linear probing quotient filters cannot support
deletions and the bounded growing technique.

\section{Concurrent Quotient Filter}
\label{sec:conc}
From a theoretical point of view using an external array of locks
should lead to a good concurrent quotient filter.  As long as the
number of locks is large enough to reduce contention ($>p^2$) and the
number of slots per lock is large enough to ensure that clusters don't
span multiple locking regions.  But there are multiple issues that
complicate the implementation in practice.  In this section we
describe the difficulties of implementing an efficient concurrent
quotient filter, attributes of a successful concurrent quotient
filter, and the concurrent variants we implemented.

\paragraph*{Problems with Concurrent Quotient Filters.}
The first consideration with every concurrent data structure should be
correctness.  The biggest problem with implementing a concurrent
quotient filter is that multiple data entries have to be read in a
consistent state for a query to succeed.  All commonly known variants
(except the linear probing quotient filter) use at least two status
bits per slot, i.e., the \emph{occupied-bit} and some kind of
\emph{run-delimiter-bit} (the run-delimiter might take different
forms).  The remainders and the run-delimiter-bit that belong to
one slot cannot reliably be stored close to their slot and its
occupied-bit. Therefore, the occupied-bit and the run-delimiter-bit
cannot be updated with one atomic operation.


For this reason, implementing a lock-free quotient filter that uses
status bits would cause significant overheads (i.e. additional memory
or multiple passes over the accessed data).  The problem is that
reading parts of multiple different but individually consistent states
of the quotient filter leads to an inconsistent perceived state.  To
show this we assume that insertions happen atomically and transfer the
table from one consistent state directly into the new final consistent
state.  During a query we have to find the canonical run
and scan the remainders within this run.  To find the canonical run we
have to compute its rank among non-empty runs using the occupied-bits
(either locally by iterating over slots or globally using a
rank-select-data structure) and then find the run with the same rank
using the run-delimiter-bits.  There is no way to guarantee that the
overall state has not changed between finding the rank of the
occupied-bit bit and finding the appropriate run-delimiter.
Specifically we might find a new run that was created by an insertion
that shifted the actual canonical run.  Similar things can happen when
accessing the remainders especially when remainders are not stored
interleaved with their corresponding status bits.  Some of these
issues can be fixed using multiple passes over the data but ABA
problems might arise in particular when deletions are possible.

To avoid these problems while maintaining performance we have two
options: either comparing all following remainders and removing the
status bits all-together (i.e. the linear probing quotient filter) or
by using locks that protect the relevant range of the table.



\paragraph*{Goals for Concurrent Quotient Filters.}
Besides the correctness there are two main goals for any concurrent
data structure: scalability and overall performance.
One major performance advantage of quotient filters over other AMQ
data structures is their cache efficiency.  Any successful concurrent
quotient filter variant should attempt to preserve this advantage.
Especially insertions into short or even empty clusters and
(unsuccessful) queries with an empty canonical run lead to
operations with only one or two accessed cache lines and at most one
write operation.  Any variant using external locks has an immediate
disadvantage here.

\paragraph*{Concurrent Slot Storage.}
\label{sec:conc_compact}
Whenever data is accessed by multiple threads concurrently, we have to
think about the atomicity of operations.  To be able to store all data
connected to one slot together we alternate between status bits and
remainders.  The slots of a quotient filter can have arbitrary sizes
that do not conform to the size of a standard atomic data type.
Therefore, we have to take into account the memory that is wasted by
using these basic data types inefficiently.  It would be common to
just use the smallest data type that can hold both the remainder and
status bits and waste any excess memory.  But quotient filters are
designed to be space efficient.  Therefore, we need a different
method.

We use the largest common atomic data type (64~bit) and pack as many
slots as possible into one data element -- we call this packed
construct a group (of slots).  This way, the waste of multiple slots
accumulates to encompass additional elements.  Furthermore, using this
technique we can atomically read or write multiple slots at once --
allowing us to update in bulk and even avoid some locking.

An alternative would have been not to waste any memory and store slots
consecutively.  However, this leads to slots that are split over the
borders of the aligned common data types and might even cross the
border between cache lines.  While non-aligned atomic operations are
possible on modern hardware, we decided against this method after a
few preliminary experiments showed how inefficient they performed in
practice.

\subsection{Concurrent Linear Probing Quotient Filter}
\label{sec:conc_lp_filter}
Operations on a linear probing quotient filter (as described in
Section~\ref{sec:var_lpqf}) can be executed concurrently similar to
operations of a linear probing hash table~\cite{growt_topc}. The table
is constructed out of merged atomic slots as described above.  Each
insertion changes the table atomically using a single compare and swap
instruction.  This implementation is even lock free, because at least
one competing write operation on any given cell is successful.
Queries find all remainders that were inserted into their canonical
slot, because they are stored between their canonical slot and the
next empty slot and the contents of a slot never change once something is stored within it.


\subsection{Concurrent Quotient Filter with Local Locking}
\label{sec:conc_local_locking}
In this section we introduce an easy way to implement concurrent
quotient filters without any memory overhead and without increasing
the number of cache lines that are accessed during operations.  This
concurrent implementation is based on the basic (3-status-bit)
quotient filter.  The same ideas can also be implemented with the
2-status-bit variant, but there are some problems discussed later in
this section.

\begin{algorithm}[htb]
  \caption{\label{alg:insert}Concurrent Insertion}
  \footnotesize
\begin{code}
  $(\quot{}, \rem{}) \leftarrow \fprint(\textit{key})$\\
  \ccomment{Block A: try a trivial insertion}\\
  $\cvar{table\_section} \leftarrow$ atomically load data around slot $\quot{}$\\
  \cif{} insertion into \cvar{table\_section} is trivial \cthen\+\\
    finish that insertion with a CAS and \creturn\-\\

  \ccomment{Block B: write lock the supercluster}\\
  \cscanr{} \crunvar{} $\leftarrow \quot{}$\+\\
    \cif{} \crunvar{} is \cemph{write locked} \cthen{}\+\\
      \cwait{} until released and \ccontinue\-\\
    \cif{} \crunvar{} is \cemph{empty} \cthen{}\+\\
      \clock{} \crunvar{} with \cemph{write lock} and \cbreak{}\+\\
        \cif{} CAS unsuccessful re-examine \crunvar{}\-\-\-\\

  \ccomment{Block C: read lock the cluster}\\
  \cscanl{} \crunvar{} $\leftarrow \quot{}$\+\\
  \cif{} \crunvar{} is \cemph{read locked} \cthen{}\+\\
    \cwait{} until released and retry this slot\-\\
  \cif{} \crunvar{} is \cemph{cluster start} \cthen{}\+\\
    \clock{} \crunvar{} with \cemph{read lock} and \cbreak{}\+\\
      \cif{} CAS unsuccessful re-examine \crunvar{}\-\-\-\\

  \ccomment{Block D: find the correct run}\\
  $\textit{occ} = 0;\ \ \textit{run} = 0$\\
  \cscanr{} \crunvar{}\+\\
    \cif{} \crunvar{} is \cemph{occupied} \cand{} \crunvar{} $<$ canonical slot \cthen{} \textit{occ}\texttt{++}\\
    \cif{} \crunvar{} is \cemph{run start} \cthen{} \textit{run}\texttt{++}\\
    \cif{} $\textit{occ} = \textit{run}$ \cand{} \crunvar{} $\geq$ canonical slot \cthen{} \cbreak{}\-\\

  \ccomment{Block E: insert into the run and shift}\\
  \cscanr{} \crunvar{}\+\\
    \cif{} \crunvar{} is \cemph{read locked} \cthen{}\+\\
      \cwait{} until released\-\\
    store \rem{} in correct slot\\
    shift content of following slots\\
    (keep groups consistent)\\
    \cbreak{} after overwriting the \cemph{write lock}\-\\
 \cunlock{} the \cemph{read lock}
\end{code}
\end{algorithm}

\begin{algorithm}[htb]
  \caption{\label{alg:contains}Concurrent Query}
  \footnotesize
\begin{code}
  $(\quot{}, \rem{}) \leftarrow \fprint(\textit{key})$\\
  \ccomment{Block F: try trivial query}\\
  $\cvar{table\_section} \leftarrow$ atomically load data around slot $\quot{}$\\
  \cif{} answer can be determined from \cvar{table\_section} \cthen\+\\
    \creturn{} this answer\-\\

  \ccomment{Block G: read lock the cluster}\\
  \cscanl{} \crunvar{} $\leftarrow \quot{}$\+\\
  \cif{} \crunvar{} is \cemph{read locked} \cthen{}\+\\
    \cwait{} until released and retry this slot\-\\
  \cif{} \crunvar{} is \cemph{cluster start} \cthen{}\+\\
    \clock{} \crunvar{} with \cemph{read lock} and \cbreak{}\+\\
      \cif{} CAS unsuccessful re-examine \crunvar{}\-\-\-\\

  \ccomment{Block H: find the correct run}\\
  $\textit{occ} = 0;\ \ \textit{run} = 0$\\
  \cscanr{} \crunvar{}\+\\
    \cif{} \crunvar{} is \cemph{occupied} \cand{} \crunvar{} $<$ canonical slot \cthen{} \textit{occ}\texttt{++}\\
    \cif{} \crunvar{} is \cemph{run start} \cthen{} \textit{run}\texttt{++}\\
    \cif{} $\textit{occ} = \textit{run}$ \cand{} \crunvar{} $\geq$ canonical slot \cthen{} \cbreak{}\-\\

  \ccomment{Block I: search remainder within the run}\\
  \cscanr{} \crunvar{}\+\\
    \cif{} $\crunvar{} = \rem{}$\+\\
      \cunlock{} the \cemph{read lock}\\
      \creturn{} contained\-\\
      \cif{} \crunvar{} is not continuation of this run\+\\
        \cunlock{} the \cemph{read lock}\\
        \creturn{} not contained
\end{code}
\end{algorithm}

\paragraph*{Using status bits for local locking.}
\label{sec:conc_status_locking}
To implement our locking scheme we use two combinations of status bits
that are impossible to naturally appear in the table (see
Table~\ref{tab:status}) -- \readlock{} and \writelock{}.  We use
\writelock{} to implement a write lock.  In the beginning of each
write operation, \writelock{} is stored in the first free slot after
the supercluster, see Block~B Algorithm~\ref{alg:insert} (using a
compare and swap operation).  Insertions wait when encountering a
write lock.  This ensures that only one insert operation can be active
per supercluster.

The combination \readlock{} is used as a read lock. All operations
(both insertions Block~C Algorithm~\ref{alg:insert} and queries
Block~G Algorithm~\ref{alg:contains}) replace the status bits of the
first element of their canonical cluster with \readlock{}.
Additionally, inserting threads wait for each encountered read lock
while moving elements in the table, see Block~E
Algorithm~\ref{alg:insert}.  When moving elements during the insertion
each encountered cluster start is shifted and becomes part of the
canonical cluster that is protected by the insertion's original read
lock.  This ensures, that no operation can change the contents
of a cluster while it is protected with a read lock.

\paragraph*{Avoiding locks.}
\label{sec:conc_avoiding_locks}
Many instances of locking can be avoided, e.g., when the canonical
slot for an insertion is empty (write the elements with a single
compare and swap operation), when the canonical slot of a query
either has no run (first status bit), or stores the wanted fingerprint.
In addition to these trivial instances of lock elision, where the
whole operation happens in one slot, we can also profit from our
compressed atomic storage scheme (see Section~\ref{sec:conc_compact}).
Since we store multiple slots together in one atomic data member,
multiple slots can be changed at once.  Each operation can act without
acquiring a lock, if the whole operation can be completed without
loading another data element.  The correctness of the algorithm is
still guaranteed, because the slots within one data element are always
in a consistent state.

Additionally, since the table is the same as it would be for the basic
non-concurrent quotient filter, queries can be executed without
locking if there can be no concurrent write operations (e.g. during an
algorithm that first constructs the table and then queries it
repeatedly).  This fact is actually used for the fully expandable variant
introduced in Section~\ref{sec:dyn_quotient}.


\paragraph*{Deletions.}
\label{sec:conc_del}
Given our locking scheme deletions could be implemented in the
following way (we did not do this since it has some impact on
insertions).  A thread executing a deletion first has to acquire a
write lock on the supercluster then it has to scan to the left (to
the canonical slot) to guarantee that there was no element removed
while finding the cluster end.  Then the read lock is acquired and the
deletion is executed similar to the sequential case.  During the
deletion it is possible that new cluster starts are created due to
elements shifting to the left into their canonical slot. Whenever this
happens, the cluster is created with read locked status bits
(\readlock{}).  After the deletion all locks are unlocked.

This implementation would impact some of the other operations.  During
a query, after acquiring the read lock when moving to the right, it
would be possible to hit a new cluster start or an empty cell before
reaching the canonical slot.  In this case the previous read lock is
unlocked and the operation is restarted.  During an insertion after
acquiring the write lock it is necessary to scan to the left to check
that no element was removed after the canonical slot.

\paragraph*{2-Status-Bit Concurrent Quotient Filter.}
In the 2-status-bit variant of the quotient filter there are no unused
status bit combinations that can be used as read or write locks.  But
zero can never be a remainder when using this variant, therefore, a
slot with an empty remainder but non-zero status bits cannot occur.
We can use such a cell to represent a lock.  Using this method,
cluster starts within a larger supercluster cannot be recognized when
scanning to the left ,i.e., in Blocks C and G (of
Algorithms~\ref{alg:insert} and \ref{alg:contains}).  Instead we can
only recognize supercluster starts (a non-empty cell to the right of
an empty cell).

To write lock a supercluster, we store \texttt{01} in the status bits
of an otherwise empty slot after the supercluster.  To read lock a
supercluster
we remove the remainder from the table and store it locally until the
lock is released (status bits \texttt{11}).  However, this concurrent
variant is not practical because we have to ensure that the
read-locked slot is still a supercluster start (the slot to its left
remains empty).  To do this we can atomically compare and swap both
slots.  However this is a problem since both slots might be stored in
different atomic data types or even cache lines.  The variant is still
interesting from a theoretical perspective where a compare and swap
operation changing two neighboring slots is completely reasonable
(usually below 64 bits).  However, we have implemented this variant
when we did some preliminary testing with non-aligned compare and swap
operations and it is non-competitive with the other implementations.

\paragraph*{Growing concurrently.}
\label{sec:conc_growing}
The bounded growing technique (described in
Section~\ref{sec:prelim_quotient}) can be used to increase the
capacity of a concurrent quotient filter similar to that of a
sequential quotient filter.  In the concurrent setting we have to
consider two things: distributing the work of the migration between
threads and ensuring that no new elements are inserted into parts of
the old table that were already migrated (otherwise they might be
lost).

To distribute the work of migrating elements, we use methods similar
to those in Maier et al.~\cite{growt_topc}.  After the migration is
triggered, every thread that starts an operation first helps with the
migration before executing the operation on the new table.  Reducing
interactions between threads during the migration is important for
performance.  Therefore, we migrate the table in blocks.  Every thread
acquires a block by incrementing a shared atomic variable.  Each block
has a constant size of 4096 slots.  The migration of each block
happens one supercluster at a time.  Each thread migrates all
superclusters that begin in its block.  This means that a thread does
not migrate the first supercluster in its block, if it starts in the
previous block.  It also means that the thread migrates elements from
the next block, if its last supercluster crosses the border to that
block.  The order of elements does not change during the migration,
because they remain ordered by their fingerprint.  In general this
means that most elements within one block of the original table are
moved into one of two blocks in the target table (block $i$ is moved
to $2i$ and $2i+1$).
By assigning clusters depending on the starting slot of their
supercluster, we enforce that there are no two threads accessing the
same slot of the target table.  Hence, no atomic operations or locks
are necessary in the target table.

As described before, we have to ensure that ongoing insert operations
either finish correctly or help with the migration before inserting
into the new table. Ongoing queries also have to finish to prevent
deadlocks.
To prevent other threads from inserting elements during the migration,
we write lock each empty slot and each supercluster before it is
migrated.  These \emph{migration-write-locks} are never released.  To
differentiate migration-write-locks from the ones used in a normal
insertions, we write lock a slot and store a non-zero remainder (write
locks are usually only stored in empty slots).  This way, an ongoing
insertion recognizes that the encountered write lock belongs to a
migration.  The inserting thread first helps with the migration before
restarting the insertion after the table is fully migrated.  Queries
can happen concurrently with the migration, because the migration does
not need read locks.

\section{Fully Expandable QFs}
\label{sec:dyn_quotient}
The goal of this fully expandable quotient filter is to offer a resizable
quotient filter variant with a bounded false positive rate that works
well even if there is no known bound to the number of elements
inserted.  Adding new fingerprint bits to existing entries is
impossible without access to the inserted elements.  We adapt a
technique that was originally introduced for scalable Bloom
filters~\cite{scalable_bloom}.  Once a quotient filter is filled, we
allocate a new \emph{additional} quotient filter.  Each subsequent
quotient filter increases the fingerprint size.  Overall, this ensures
a bounded false positive rate.

We show that even though this is an old idea, it offers some
interesting possibilities when applied to quotient filters, i.e.,
avoiding locks on lower levels, growing each level using the bounded
growing technique, higher fill degree through cascading inserts, and
early rejection of queries also through cascading inserts.

\subsection{Architecture}
\label{sec:dyn_arch}
This new data structure starts out with one quotient filter, but over
time, it may contain a set of quotient filters we call levels.  At any
point in time, only the newest (highest) level is active.  Insertions
operate on the active level.  The data structure is initialized with
two user-defined parameters the \emph{expected capacity} $c$ and the
upper \emph{bound for the false positive rate} $\fpratedach{}$.  The
first level table is initialized with $m_0$ slots where
$m_0 = 2^{q_0}$ is the first power of 2 where
$\filldeg{}_{\textit{grow}}\cdot m_0$ is larger than $c$, here
$\filldeg{}_{\textit{grow}}$ is the fill ratio where growing is
triggered and $\maxline{n_i} = \filldeg{}_{\textit{grow}}\cdot m_i$ is
the maximum number of elements level $i$ can hold.  The number of
remainder bits $r_0$ is chosen such that
$\fpratedach{} > 2\filldeg{}_{\textit{grow}}\cdot 2^{-r_0}$.  This
first table uses $k_0=q_0+r_0$ bits for its fingerprint.

Queries have to check each level.  Within the lower levels queries do
not need any locks, because the elements there are finalized.  The query
performance depends on the number of levels.  To keep the number of
levels small, we have to increase the capacity of each subsequent
level. To also bound the false positive rate, we have to reduce the
false positive rate of each subsequent level.  To achieve both of
these goals, we increase the size of the fingerprint $k_i$ by two for
each subsequent level ($k_i = 2+k_{i-1}$).  Using the longer
fingerprint we can ensure that once the new table holds twice as many
elements as the old one ($\maxline{n_i} = \maxline{n_{i+1}}/2$), it
still has half the false positive rate
($\fprateidach{i} = \maxline{n_i}\cdot 2^{-k_i} = 2 \fprateidach{i+1}
= 2\cdot \maxline{n_{i+1}} \cdot 2^{-k_{i+1}}$).

When one level reaches its maximum capacity $\maxline{n_i}$ we
allocate a new level.  Instead of allocating the new level to
immediately have twice the number of slots as the old level, we
allocate it with one 8th of the final size, and use the bounded
growing algorithm (described in Section~\ref{sec:prelim_quotient}) to grow it
to its final size (three growing steps).  This way, the table has a
higher average fill rate (at least
$2/3\cdot\filldeg{}_{\textit{grow}}$ instead of
$1/3\cdot\filldeg{}_{\textit{grow}}$ \todolater{for future reference
  4th of the size would have been enough}).



\begin{theorem}[Bounded $\fprate$ in expandable QF]
  \label{theo:dyn_fprate}
  The fully expandable quotient filter holds the false positive
  probability $\fpratedach{}$ set by the user independently
  of the number of inserted elements.
\end{theorem}

\begin{proof}
  For the following analysis, we assume that fingerprints can
  potentially have an arbitrary length.
  The analysis of the overall false positive rate $\fprate{}$ is very
  similar to that of the scalable Bloom filter.  A false positive
  occurs if one of the $\ell$ levels has a false positive
  $\fprate = 1-\prod_i(1-\fpratei{i})$.  This can be approximated with
  the Weierstrass inequality $\fprate \leq \sum_{i=1}^\ell\fpratei{i}$.
  When we insert the shrinking false positive rates per level
  ($\fpratei{i+1} = \fpratei{i}/2$) we obtain a geometric sum which is
  bounded by $2\fpratei{1}$:
  $\sum_{i=0}^{\ell-1} \fpratei{i} 2^{-i} \leq
  2\fpratei{1}<\fpratedach{}$.
\end{proof}

Using this growing scheme the number of filters is in $O(\log{n/T})$,
therefore, the bounds for queries are similar to those in a broad tree
data structure.  But due to the necessary pointers tree data
structures take significantly more memory.  Additionally, they are
difficult to implement concurrently without creating contention on the
root node.

\subsection{Cascading Inserts}
\label{sec:dyn_quick_insert}
The idea behind cascading inserts is to insert elements on the lowest
possible level.  If the canonical slot in a lower level is empty, we
insert the element into that level.  This can be done using a simple
compare and swap operation (without acquiring a write lock).  Queries
on lower levels can still proceed without locking, because insertions
cannot move existing elements.

The main reason to grow the table before it is full is to improve the
performance by shortening clusters.  The trade-off for this is space
utilization.  For the space utilization it would be optimal to fill
each table 100\%.  Using cascading inserts this can be achieved while
still having a good performance on each level.  Queries on the
lower level tables have no significant slow down due to cascading
inserts, because the average cluster length remains small (cascading
inserts lead to one-element clusters).  Additionally, if we use
cascading inserts, we can abort queries that encounter an empty
canonical slot in one of the lower level tables, because this slot
would have been filled by an insertion.

Cascading inserts do incur some overhead.  Each insertion checks every
level whether its canonical slot is empty.  However, in some
applications checking all levels is already necessary.  For example in
applications like element unification all lower levels are checked to
prevent repeated insertions of one element.  Applications like this
often use combined query and insert operations that only insert if
the element was not yet in the table.  Here cascading inserts do not
cost any overhead.

\section{Experiments}
\label{sec:exp}
The performance evaluations described in this section consist of two
parts: Some tests without growing where we compare our implementation
of concurrent quotient filters to other AMQ data structures, and some
tests with dynamically growing quotient filter implementations.  There
we compare three different variants of growing quotient filters:
bounded growing, our fully expandable multi level approach, and the multi
level approach with cascading inserts.

All experiments were executed on a two socket Intel Xeon E5-2683 v4
machine with 16 cores per socket, each running at 2.1~GHz with 40MB
cache size, and 512~GB main memory.
The tests were compiled using gcc 7.4.0 and the operating system is
Ubuntu 18.04.02.  Each test was repeated 9 times using 3 different
sequences of keys (same distribution) -- 3 runs per sequence.

\paragraph*{Non-Growing Competitors}
\label{sec:exp_nongrow}
\begin{itemize}
\item[\sqrfull] Linear probing quotient filter presented in
  Section~\ref{sec:conc_lp_filter}
\item[\cirfull] Locally locked quotient filter presented in
  Section~\ref{sec:conc_local_locking} (using status bits as locks)
\item[\diafull] Externally locked quotient filter using an array with one lock
  per 4096 slots
\item[\cqfmarker] Counting quotient filter -- implementation by Pandey
  et al.~\cite{counting_qf, counting_qf_git}
\item[\trifull] Bloom filter. Our implementation uses the same amount
  of memory as the quotient filters we compared to ($m\cdot (r+3)$
  bits) and only \doneupdate{4} hash functions ($\Rightarrow$ better
  performance at a false positive rate comparable to our quotient
  filters).
\end{itemize}

\paragraph*{Speedup.}
\begin{figure*}[!htb]
  \centering
  \includegraphics[scale=.68]{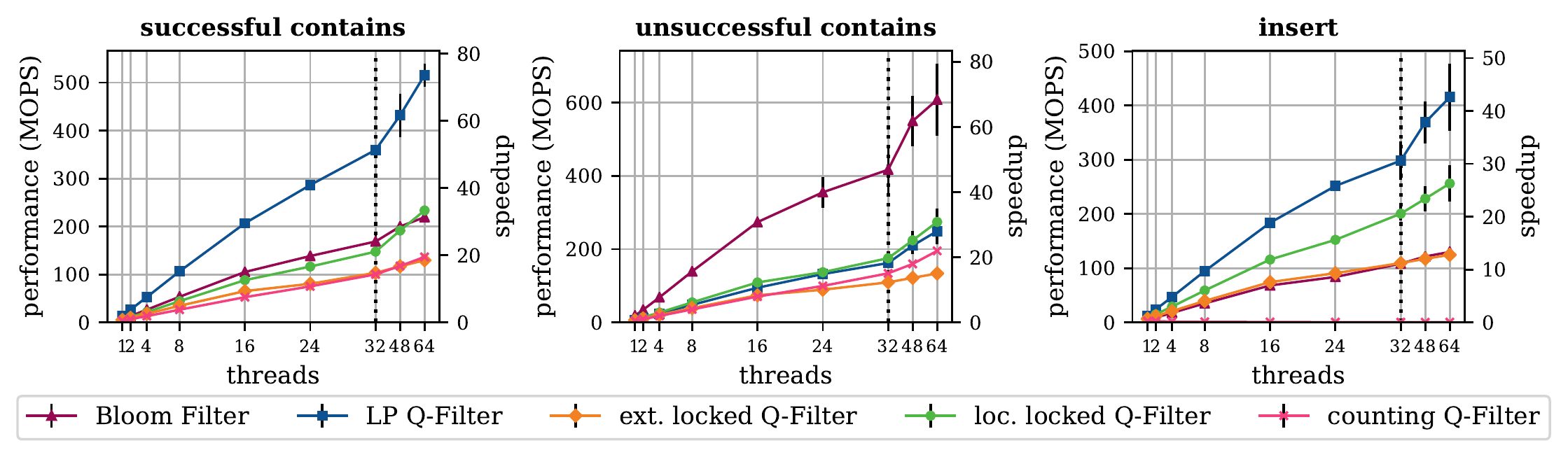}
  \caption{\label{fig:speedup} Throughput over the number of threads.
    Strong scaling measurement inserting 24M elements into a table
    with $2^{25} \approx 33.6\text{M}$ slots ($72\%$ fill degree).
    The speedup is measured compared to a sequential basic quotient
    filter.  \emph{Note:} the x-axis does not scale linearly past 32
    threads.}
\end{figure*}

The first test presented in Figure~\ref{fig:speedup} shows the
throughput of different data structures under a varying number of
operating threads.  We fill the table with 24M uniformly random elements
through repeated insertions by all threads into a previously empty
table of $2^{25}\approx 33.6\text{M}$ slots ($\filldeg \approx 72\%$)
using $r=\doneupdate{10}$ remainder bits (\doneupdate{13} for the linear probing
quotient filter).  Then we execute 24M queries with uniformly random keys
(none of which were previously inserted) this way we measure the
performance of unsuccessful queries and the rate of false positives
(not shown since it is independent of $p$).  At last, we execute one
query for each inserted element. In each of these three phases we
measure the throughput dependent on the number of processors.  The
base line for the speedups was measured using a sequential quotient
filter, executing the same test.

We can see that all data structures scale close to linearly with the
number of processors.  There is only a small bend when the number of
cores exceeds the first socket and once it reaches the number of
physical cores.  But, the data structures even seem to scale
relatively well when using hyperthreading.  The absolute performance
is quite different between the data structures.  In this test the
linear probing quotient filter has by far the best throughput
(excluding the unsuccessful query performance of the Bloom filter).
Compared to the locally locking variant it performs \doneupdate{1.5}
times better on inserts and \doneupdate{2.4} times better on
successful queries (speedups of \doneupdate{30.6 vs. 20.5} on inserts
and \doneupdate{51.3 vs 21.0} on successful queries at $p=32$).  Both
the externally locking variant performs far worse (speedups below
\doneupdate{12} for insertions).  Inserting into the counting quotient
filter does not scale with the number of processors at all, instead,
it starts out at about 2M inserts per second and gets worse from
there.  Both variants also perform worse on queries with a speedup of
below \doneupdate{15} ($p=32$).  This is due to the fact that each
operation incurs multiple cache faults -- one for locking and one for
the table access.  In the case of the counting quotient filter there
is also the additional work necessary for updating the rank-select
bitvectors.  The concurrent Bloom filter has very different
throughputs depending on the operation type.  It has by far the best
throughput on unsuccessful queries (speedup of \doneupdate{46.8} at
$p=32$).  The reason for this is that the described configuration with
only \doneupdate{4} hash functions and more memory than a usual bloom
filters (same memory as quotient filters) leads to a sparse filter,
therefore, it most unsuccessful queries can be aborted after one or
two cache misses.

\paragraph*{Fill ratio.}
\begin{figure*}[!htb]
  \centering
  \includegraphics[scale=.68]{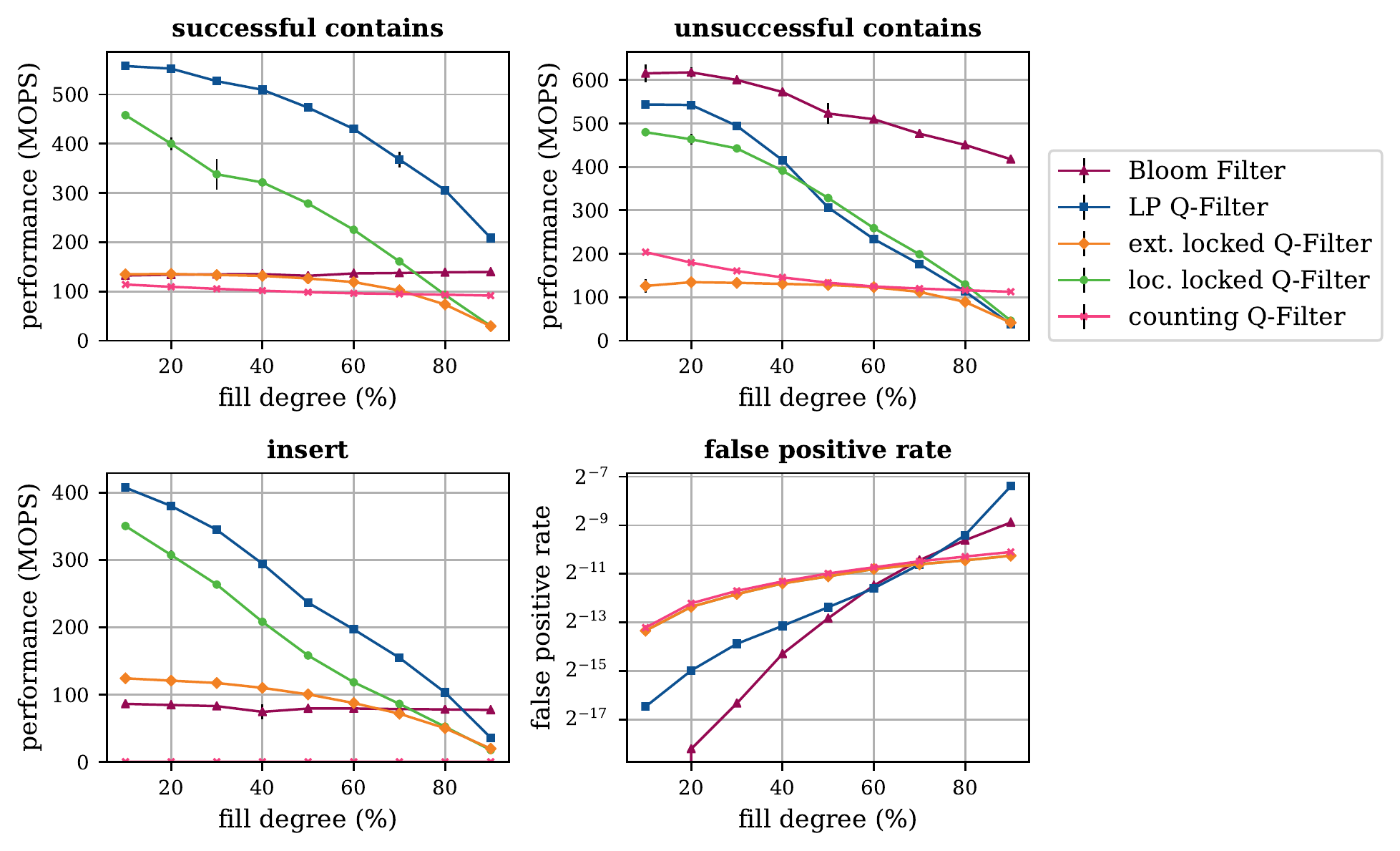}
  \caption{\label{fig:fill_degree} Throughput over fill degree.
    Using $p=32$ threads and $2^{25}\approx 33.6\text{M}$ slots.  Each
    point of the measurement was tested using 100k operations.}
\end{figure*}

In this test we fill a table with $2^{25} \approx 33.6\text{M}$ slots
and $r=\doneupdate{10}$.  Every 10\% of the fill ratio, we execute a
performance test.  During this test each table operation is executed
100k times.  The results can be seen in Figure~\ref{fig:fill_degree}.
The throughput of all quotient filter variants decreases steadily with
the increasing fill ratio, but the relative differences remain close
to the same.  Only the counting quotient filter and the Bloom filter
have running times that are somewhat independent of the fill degree.
On lower fill degrees the linear probing quotient filter and the
locally locked quotient filter display their strengths. With about
\doneupdate{236\% and 157\%} higher insertion throughputs than the
external locking quotient filter, and \doneupdate{358\% and 210\%}
higher successful query throughputs than the Bloom filter at $50\%$
fill degree.

\paragraph*{Growing Benchmark}
\label{sec:exp_grow}

\begin{figure*}[!htb]
  \centering
  \includegraphics[scale=.68]{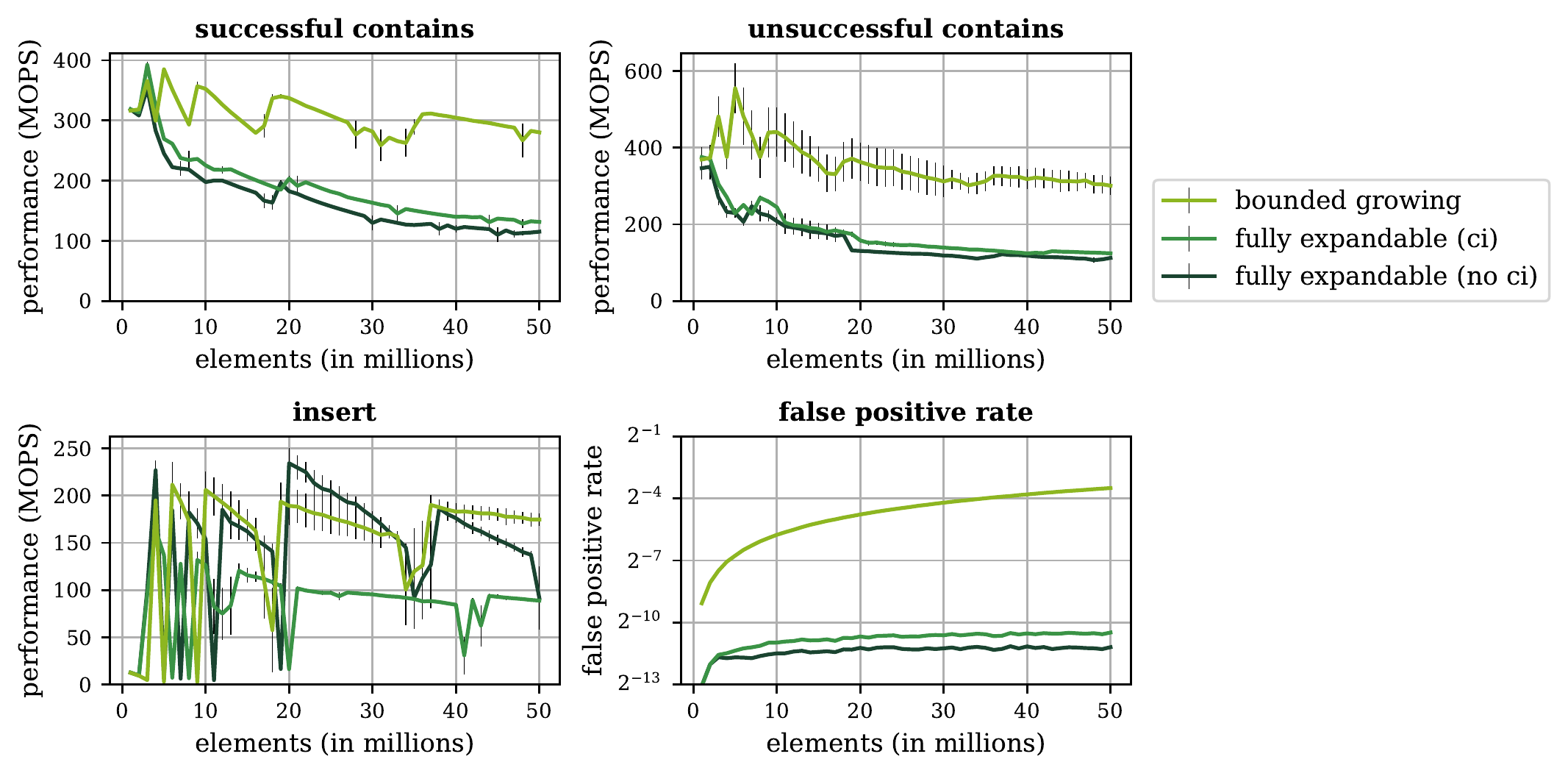}
  \caption{\label{fig:progression} Throughput of the growing tables (all variants based on the local
    locking quotient filter). Using $p=32$ threads 50 segments of 1M elements are inserted
    into a table initialized with $2^{19} \approx 520\text{k}$ slots. }
\end{figure*}
In the benchmark shown in Figure~\ref{fig:progression}, we insert 50M
elements into each of our dynamically sized quotient filters bounded
growing, fully expandable without cascading inserts (no ci), and fully
expandable with cascading inserts (ci) (all based on the quotient filter
using local locking).  The filter was initialized with a capacity of
only $2^{19}\approx 520\text{k}$ slots and a target false positive
rate of $\fpratedach = 2^{-10}$.  The inserted elements are split into 50
segments of 1M elements each.  We measure the running time of
inserting each segment as well as the query performance after each
segment (using 1M successful and unsuccessful queries).

As expected, the false positive rate grows linearly when we only use
bounded growing, but each query still consists of only one table
lookup -- resulting in a query performance similar to that of the
non-growing benchmark.  Both fully expandable variants have a bounded
false positive rate.  But their query performance suffers due to the
lower level look ups and the additional cache lines that have to
be accessed with each additional table.  Cascading inserts can improve
query times by \doneupdate{13\%} for successful queries and
\doneupdate{12\%} for unsuccessful queries (averaged over all
segments), however, slowing down insertions significantly.

\section{Conclusion}
\label{sec:dis}
\paragraph*{Future Work.}
\label{sec:dis_future}
%
%
We already mentioned a way to implement deletions in our
concurrent table Section~\ref{sec:conc_del}.  Additionally, one would
probably want to add the same counting method used by Pandey et
al.~\cite{counting_qf} this should be fairly straight forward, since
the remainders in our implementation are also stored in increasing
order.

One large opportunity for improvement could be to use the local
locking methods presented here as a backoff mechanism for a solution
that uses hardware transactional memory. Transactional memory could
also improve the memory usage by removing the overheads introduced by
the grouped storage method we proposed.  Such an implementation might
use unaligned compare and swap operations, but only if a previous
transaction failed (this would hopefully be very rare).

The fully expandable quotient filter could be adapted to specific use
cases, e.g., one could use the normal insertion algorithm to fill a
fully expandable quotient filter.  Then one could scan all tables (in
parallel) to move elements from later levels into earlier levels, if
their slots are free.  The resulting data structure could then be
queried as if it was built with cascading inserts.



\paragraph*{Discussion.}
In this publication, we have shown a new technique for concurrent
quotient filters that uses the status bits inherent to quotient
filters for localized locking.  Using this technique, no
additional cache lines are accessed (compared to a sequential quotient
filter).  We were able to achieve a \doneupdate{1.8} times increase in
insert performance over the external locking scheme (\doneupdate{1.4}
on queries; both at $p=32$).  Additionally, we proposed a simple
linear probing based filter that does not use any status bits and is
lock-free.  Using the same amount of memory, this filter achieves even
better false positive rates up to a fill degree of 70\% and also
\doneupdate{1.5} times higher insertion speedups (than the local
locking variant).

We also propose to use the bounded growing technique available in
quotient filters to refine the growing technique used in scalable
Bloom filters.  Using this combination of techniques guarantees that
the overall data structure is always at least
$2/3\cdot\filldeg{}_{\textit{grow}}$ filled (where
$\filldeg{}_{\textit{grow}}$ is the fill degree where growing is
triggered). Using cascading inserts, this can even be improved by
filling lower level tables even further, while also improving query
times by over \doneupdate{12\%}.

Our tests show that there is no optimal AMQ data structure.  Which
data structure performs best depends on the use case and the expected
workload.  The linear probing quotient filter is very good, if the
table is not densely filled.  The locally locked quotient filter is
also efficient on tables below a fill degree of 70\%.  But, it is also
more flexible for example when the table starts out empty and is
filled to above 70\% (i.e., constructing the filter).  Our growing
implementations work well if the number of inserted elements is not
known prior to the table's construction.  The counting quotient filter
performs well on query heavy workloads that operate on densely filled
tables.

\FloatBarrier
\bibliography{quotient_paper.bib}{}

\begin{thebibliography}{10}

\bibitem{qf_network}
Mohammad Al-hisnawi and Mahmood Ahmadi.
\newblock {D}eep {P}acket {I}nspection using {Q}uotient {F}ilter.
\newblock {\em IEEE Communications Letters}, 20(11):2217--2220, Nov 2016.

\bibitem{scalable_bloom}
Paulo~S{\'e}rgio Almeida, Carlos Baquero, Nuno Pregui\c{c}a, and David
  Hutchison.
\newblock {S}calable {B}loom {F}ilters.
\newblock {\em Inf. Process. Lett.}, 101(6):255--261, March 2007.

\bibitem{dont_thrash}
Michael~A. Bender, Martin Farach-Colton, Rob Johnson, Russell Kraner,
  Bradley~C. Kuszmaul, Dzejla Medjedovic, Pablo Montes, Pradeep Shetty,
  Richard~P. Spillane, and Erez Zadok.
\newblock {D}on't {T}hrash: How to cache your hash on flash.
\newblock {\em Proc. VLDB Endow.}, 5(11):1627--1637, July 2012.

\bibitem{morton_filter}
Alex~D. Breslow and Nuwan~S. Jayasena.
\newblock Morton filters: fast, compressed sparse cuckoo filters.
\newblock {\em The VLDB Journal}, Aug 2019.

\bibitem{robin}
Pedro Celis, Per-Ake Larson, and J~Ian Munro.
\newblock Robin {H}ood {H}ashing.
\newblock In {\em 26th Annual Symposium on Foundations of Computer Science
  (FOCS)}, pages 281--288. IEEE, 1985.

\bibitem{Cle84}
J.~G. Cleary.
\newblock Compact hash tables using bidirectional linear probing.
\newblock {\em IEEE Transactions on Computers}, C-33(9):828--834, 1984.

\bibitem{xxhash}
Yan Collet.
\newblock xx{H}ash.
\newblock \url{https://github.com/Cyan4973/xxHash}.
\newblock Accessed March 21, 2019.

\bibitem{cuckoo_filter}
Bin Fan, Dave~G. Andersen, Michael Kaminsky, and Michael~D. Mitzenmacher.
\newblock {C}uckoo {F}ilter: Practically better than {B}loom.
\newblock In {\em 10th ACM International on Conference on Emerging Networking
  Experiments and Technologies}, pages 75--88, 2014.

\bibitem{gpu_qf}
Afton Geil, Martin Farach-Colton, and John~D. Owens.
\newblock {Q}uotient {F}ilters: Approximate membership queries on the {GPU}.
\newblock In {\em 2018 IEEE International Parallel and Distributed Processing
  Symposium (IPDPS)}, pages 451--462, May 2018.

\bibitem{knuth}
Donald~Ervin Knuth.
\newblock {\em The art of computer programming: sorting and searching},
  volume~3.
\newblock Pearson Education, 1997.

\bibitem{concurrent_cuckoo}
Xiaozhou Li, David~G. Andersen, Michael Kaminsky, and Michael~J. Freedman.
\newblock Algorithmic {I}mprovements for {F}ast {C}oncurrent {C}uckoo
  {H}ashing.
\newblock In {\em Proceedings of the 9th European Conference on Computer
  Systems}, EuroSys '14, pages 27:1--27:14, New York, NY, USA, 2014. ACM.

\bibitem{growt_topc}
Tobias Maier, Peter Sanders, and Roman Dementiev.
\newblock Concurrent {H}ash {T}ables: Fast and general(?)!
\newblock {\em ACM Trans. Parallel Comput.}, 5(4):16:1--16:32, February 2019.

\bibitem{adaptive_cuckoo_filter}
Michael Mitzenmacher, Salvatore Pontarelli, and Pedro Reviriego.
\newblock {A}daptive {C}uckoo {F}ilters.
\newblock In {\em 2018 Proceedings of the Twentieth Workshop on Algorithm
  Engineering and Experiments (ALENEX)}, pages 36--47. SIAM, 2018.

\bibitem{lockfree_cuckoo}
Nhan Nguyen and Philippas Tsigas.
\newblock Lock-free cuckoo hashing.
\newblock In {\em 34th IEEE International Conference on Distributed Computing
  Systems (ICDCS)}, pages 627--636, June 2014.

\bibitem{counting_qf_git}
Prashant Pandey.
\newblock Counting quotient filter.
\newblock \url{https://github.com/splatlab/cqf}.
\newblock Accessed August 07, 2019.

\bibitem{qf_bio1}
Prashant Pandey, Fatemeh Almodaresi, Michael~A. Bender, Michael Ferdman, Rob
  Johnson, and Rob Patro.
\newblock {M}antis: A fast, small, and exact large-scale sequence-search index.
\newblock {\em Cell Systems}, 7(2):201 -- 207.e4, 2018.

\bibitem{counting_qf}
Prashant Pandey, Michael~A. Bender, Rob Johnson, and Rob Patro.
\newblock A {G}eneral-{P}urpose {C}ounting {F}ilter: Making every bit count.
\newblock In {\em ACM Conference on Management of Data (SIGMOD)}, pages
  775--787, 2017.

\end{thebibliography}
\bibliographystyle{plain}
\end{document}